\documentclass[10pt,twocolumn,twoside]{IEEEtran}

\usepackage{amssymb}
\usepackage{amsthm}
\usepackage{latexsym}
\usepackage{verbatim}
\usepackage[pdftex]{graphicx}
\usepackage{fixltx2e}
\usepackage{array}
\usepackage{booktabs}
\usepackage[square, numbers,sort&compress]{natbib}
\usepackage[cmex10]{amsmath}
\usepackage{setspace,algorithm,algorithmic,enumerate}

\newtheorem{theorem}{Theorem}
\newtheorem{proposition}{Proposition}

\newtheorem{corollary}{Corollary}

\newtheorem{definition}{Definition}

\newtheorem{problem}{Problem}

\newtheorem{remark}{Remark}

{\begin{list}               
    {$\bullet$ \hfill}{
        \setlength{\leftmargin}{\parindent}
        \setlength{\parsep}{0.04\baselineskip}
        \setlength{\itemsep}{0.5\parsep}
        \setlength{\labelwidth}{\leftmargin}
        \setlength{\labelsep}{0em}}
    }
{\end{list}}

\providecommand{\cref}[1]{Chapter~\ref{chap:#1}}

\providecommand{\R}{\ensuremath{\mathbb{R}}}

\providecommand{\Z}{\ensuremath{\mathbb{Z}}}

\providecommand{\set}[1]{\left\{#1\right\}}

\providecommand{\rank}{\mathop{\mathrm{rank}}}

\renewcommand{\vec}[1]{\ensuremath{\boldsymbol{#1}}}
\providecommand{\mat}[1]{\ensuremath{\boldsymbol{#1}}}

\providecommand{\calF}{\mathcal{F}}

\providecommand{\calX}{\mathcal{X}}
\providecommand{\calY}{\mathcal{Y}}
\providecommand{\mA}{\mat{A}} 
\providecommand{\mC}{\mat{C}}

 \providecommand{\mP}{\mat{P}} 
\providecommand{\mQ}{\mat{Q}}

\providecommand{\mY}{\mat{Y}}

 \providecommand{\vk}{\vec{k}}

 \providecommand{\vN}{\vec{N}} 

 \providecommand{\vp}{\vec{p}}

\providecommand{\vu}{\vec{u}} 

\providecommand{\vx}{\vec{x}} \providecommand{\vy}{\vec{y}}
 \providecommand{\vi}{\vec{i}}

\usepackage[pdftex]{graphicx}
\graphicspath{{./TSP/}} \usepackage[cmex10]{amsmath}
\usepackage{setspace,enumerate,algorithm}
\newcommand{\ComplexityFont}[1]{%
{\ensuremath{\complexity@possiblymakesmaller{\complexity@fontcommand{#1}}}}
} \renewcommand{\ComplexityFont}[1]{{\ensuremath{\mathsf{#1}}}}

\usepackage{float}

\newfloat{alg}{tbh}{losf}
\floatname{alg}{Algorithm}

\usepackage{array}

\usepackage[caption=false,font=footnotesize]{subfig}

\usepackage{fixltx2e}

\usepackage{url} \usepackage{booktabs}

\newcommand{\figdir}{./TSP/}

\begin{document}

\title{Phase Retrieval for Sparse Signals:\\Uniqueness Conditions}

\author{Juri Ranieri,~\IEEEmembership{Student Member,~IEEE,} Amina
  Chebira,~\IEEEmembership{Member,~IEEE,} Yue M.
  Lu,~\IEEEmembership{Senior Member,~IEEE,}
  and~Martin~Vetterli,~\IEEEmembership{Fellow,~IEEE.}
  \thanks{Juri Ranieri and Martin Vetterli are with the School of
    Computer and Communication Sciences, Ecole Polytechnique
    F\'ed\'erale de Lausanne (EPFL), Lausanne, Switzerland (e-mails:
    {juri.ranieri, martin.vetterli}@epfl.ch). Amina Chebira is with
    the Swiss Center for Electronics and Microtechnology (CSEM),
    Neuch{\^ a}tel, Switzerland (e-mail: amina.chebira@csem.ch). Yue
    M. Lu is with the School of Engineering and Applied Sciences, Harvard
    University, Cambridge, MA 02138, USA (e-mail:
    yuelu@seas.harvard.edu).}
  }


\maketitle

\begin{abstract} In a variety of fields, in particular those involving
  imaging and optics, we often measure signals whose phase is missing
  or has been irremediably distorted. Phase retrieval attempts the
  recovery of the phase information of a signal from the magnitude of
  its Fourier transform to enable the reconstruction of the original
  signal. A fundamental question then is: \emph{``Under which
    conditions can we uniquely recover the signal of interest from its
    measured magnitudes?''}

  In this paper, we assume the measured signal to be sparse. This is a
  natural assumption in many applications, such as X-ray
  crystallography, speckle imaging and blind channel estimation. In
  this work, we derive a sufficient condition for the uniqueness of
  the solution of the phase retrieval (PR) problem for both discrete
  and continuous domains, and for one and multi--dimensional domains.
  More precisely, we show that there is a strong connection between PR
  and the \emph{turnpike problem}, a classic combinatorial problem. We
  also prove that the existence of \emph{collisions} in the
  autocorrelation function of the signal may preclude the uniqueness
  of the solution of PR. Then, assuming the absence of collisions, we
  prove that the solution is almost surely unique on 1--dimensional
  domains. Finally, we extend this result to multi--dimensional
  signals by solving a set of 1--dimensional problems. We show that
  the solution of the multi-dimensional problem is unique when the
  autocorrelation function has no collisions, significantly improving
  upon a previously known result. 
\end{abstract}


\section{Introduction}

\IEEEPARstart{I}n many real-world scenarios, we naturally measure
the Fourier transform (FT) of a signal of interest instead of the
signal itself. During the measuring process, it may happen that the
phase of the FT is lost or irremediably distorted. The recovery of the
phase is fundamental to reconstructing the signal, and this recovery
process is known as \emph{phase retrieval} (PR). Phase loss problems
occur in many scientific fields, particularly those involving optics
and communications. For example, in X-ray crystallography, the
measurements are the diffraction patterns of a crystallized molecule
and we would like to recover the molecule itself.

Although the PR problem has a long history with a rich literature
\cite{Millane:1990pt, Fienup:1982va}, there are still open questions
regarding the uniqueness of the solution and the existence of reliable
algorithms to recover the signal. We underline that in many
applications the signal of interest is sparse: for example, the atoms
of a molecule are distinct elements in the spatial domain. However, a
review of the main results related to PR reveals that sparsity has
only been rarely exploited to obtain uniqueness conditions or
efficient reconstruction algorithms.  Moreover, we note that the
problem is usually defined on the discrete domain for simplicity,
while most of the applications involve continuous signals.

In this paper, we present a uniqueness condition for 1--dimensional
sparse signals exploiting a previous result related to the
\emph{turnpike problem}, a classic combinatorial problem. We show that
the same uniqueness condition holds for multidimensional signals. Note
that these results are valid both for discrete and continuous domains,
the condition relies solely on the sparsity and on the characteristics
of the support of the signal. The only difference between the discrete
and the continuous problem is the probability of satisfying the
uniqueness condition.

In what follows, we set the notation and we precisely state the PR
problem for continuous signals. We then describe a number of
applications, emphasizing the role of sparsity in these scenarios, and
we show that this property can be further exploited to obtain
a uniqueness condition.

\section{Problem statement and applications}

In this section, we state the phase retrieval for signals defined on
$D$--dimensional continuous domains. We underline the difficulties
characterizing these problems and we define the non-trivial concept of
\emph{unique solution}. We introduce a sparse model for continuous
signals and we present a number of applications that can exploit such
model.  Unless otherwise stated, we use the following notation:
\begin{itemize}
\item bold lower case symbols, such as $\vx$, for vectors,
\item bold capital symbols, such as $\mA$, for matrices,
\item $x_n$ is the $n$-th element of the vector $\vx$, $A_{m,n}$ is
the element in the $m$-th row and the $n$-column of $\mA$,
\item capital calligraphic letters, such as $\calX$, for sets.
\end{itemize}

\subsection{PR on continuous domains}

Consider a $D$--dimensional continuous real-valued signal $f(\vx):\R^D
\to \R$, where $\vx$ is the position vector in the spatial
domain. While all the presented work is focused on continuous signals,
our analysis is also valid for discrete signals, it is then sufficient to
restrict the domain of the function $f(\vx)$ to $\Z^D$.

We define the FT of the signal $f(\vx)$ as
\begin{align} \widehat{f}(\vu)\stackrel{\text{\tiny
def}}{=}\int_{\R^D}
f(\vx)\exp(-j2\pi\left<\vu,\vx\right>)d\vx, \nonumber 
\end{align} where $\vu\in\R^{D}$ is the position vector in the Fourier
domain, $\left<\cdot,\cdot\right>$ is the inner product between two
vectors and the hat indicates the FT. If
we measure $\widehat{f}(\vu)$, the signal can be directly recovered by
the inverse FT (IFT), 
\begin{align} {f}(\vx)\stackrel{\text{\tiny def}}{=}\int_{\R^D}
\widehat{f}(\vu)\exp(j2\pi\left<\vu,\vx\right>)d\vu. \nonumber
\end{align} We can represent the FT in polar form as
\begin{align}
\widehat{f}(\vu)=|\widehat{f}(\vu)|\exp\left
  (j\phi(\vu)\right), \nonumber
\end{align} where $\phi(\vu)$ is the phase of the FT. Then, we define
the PR problem as follows: given the magnitude $|\widehat{f}(\vu)|$ of
the FT, recover the original signal $f(\vx)$. Since the FT is a
bijective mapping, the problem is equivalent to recovering the phase
term $\phi(\vu)$, hence the name phase retrieval. It is easy to show
that the knowledge of $|\widehat{f}(\vu)|$ is equivalent to the
knowledge of the autocorrelation function (ACF), defined as
\begin{align} a(\vx)\stackrel{\text{\tiny
def}}{=}\int_{\R^D}f(\vy)f^*(\vx+\vy)d\vy. \nonumber
\end{align} More precisely, the ACF is the IFT of
$|\widehat{f}(\vu)|^2$. Now, we have all the ingredients to state the
PR problem for a continuous signal $f(\vx)$.

\begin{figure*}[h!]
   \centering
   \subfloat{\includegraphics[scale=0.85]{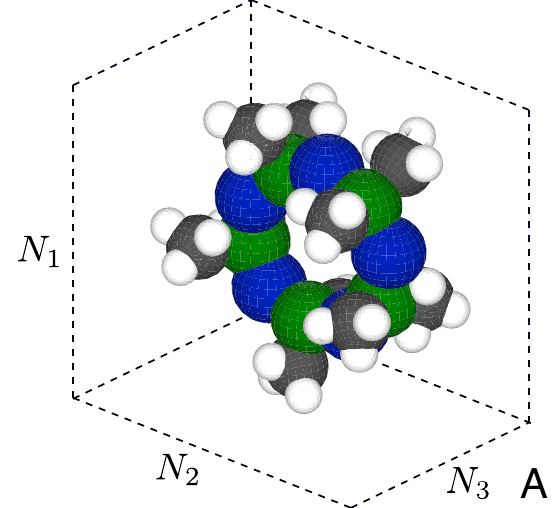}}
   \qquad
   \subfloat{\includegraphics[scale=0.85]{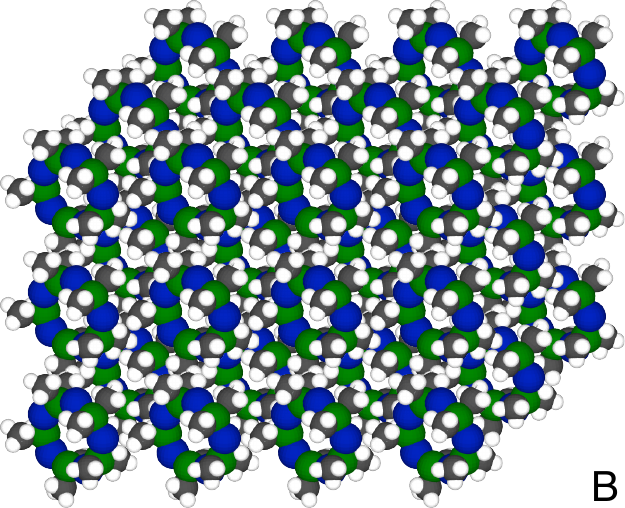}}  
   \qquad
   \subfloat{\includegraphics[scale=0.85]{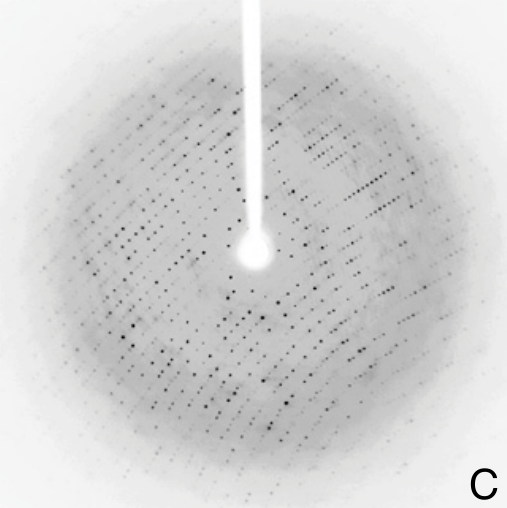}}        
   \caption[]{ Graphical representation of a unit cell (A), of a
     crystal (B) and of a diffraction pattern (C). Note that the
     crystal is simply a periodic repetition of the unit cell. In (C),
     one can see the sparsity of the diffraction pattern, reflecting
     the sparsity of the original crystal. }\vspace{-5mm}
  \label{fig:crystal}
\end{figure*}
\begin{problem}
\label{prob:pr_c} {\bf PR for continuous signals} Given the magnitude
$|\widehat{f}(\vu)|$ of the FT or the ACF $a(\vx)$ of a signal of
interest $f(\vx)$, recover the signal itself.
\end{problem}

We can show that in the general case the solution of Problem
\ref{prob:pr_c} is not unique. In fact, we can simply assign a random
phase to the measured magnitudes. Moreover, some information on the
signal $f(\vx)$ is entirely embedded into the phase of its FT and we
cannot hope to recover this information from the magnitude
only. Namely, the following transformations of $f(\vx)$,
\begin{align} f(-\vx), -f(\vx), f(\vx-\boldsymbol{\tau}), \nonumber
\end{align} do not influence the magnitude of the FT. Hence,
time--reversal, sign change and absolute position cannot be recovered
once the phase is lost.  It is then appropriate to define what ``unique'' means
with the following equivalence class
\begin{align} f(\vx)\sim g(\vx), \quad\text{if}\;f(\vx)=\pm g(\vk \pm
\vx).
\label{eq:class}
\end{align} for any $\vk\in\R^D$. Then, we say that a PR problem has a unique
solution if all the solutions are in the same equivalence class.

\subsection{Sparse signals}
A natural question arises from Problem \ref{prob:pr_c}: ``What are the
conditions that $f(\vx)$ must satisfy to have a unique PR?''

\noindent In this paper, we constrain the PR problem using a sparse
model for $f(\vx)$. More precisely, we define a $N$--sparse signal
using the Dirac delta notation, that is
\begin{equation} 
f(\vx)=\sum_{n=1}^N c^{(n)} \delta(\vx-\vx^{(n)}), 
\label{eq:model}
\end{equation} 
where the $n$-th delta has coefficient $c^{(n)}$ and is located at
$\vx^{(n)}$ and $N$ is finite. Then, the ACF is defined by the
following linear combination of $N^2-N+1$ deltas,
\begin{align} a(\vx)&=\sum_{n=1}^N\sum_{m=1}^N c^{(n)} c^{(m)}
\delta(\vx-(\vx^{(m)}-\vx^{(n)})) \nonumber\\ &=\sum_{n=0}^{N^2-N}d^{(n)}\delta(\vx-\vy^{(n)}),
\label{eq:ACF_sparse}
\end{align} where $d^{(n)}$ and $\vy^{(n)}$ are the coefficients and
the locations of the deltas in the ACF.  Note that the ACF is
centro-symmetric, meaning that for every delta located at $\vy^{(n)}$,
there is another one with the same coefficient located at
$-\vy^{(n)}$. Then, we can rewrite the ACF as

\begin{align}
a(\vx)=\sum_{n=1}^{L}d^{(n)}\delta(\vx+\vy^{(n)})+\sum_{n=0}^{L}d^{(n)}\delta(\vx-\vy^{(n)}),
\label{eq:ACF_sparse_L}
\end{align} 

where we can consider only the second sum instead of the whole ACF,
since it contains all the available information. 

\subsection{Applications}
\label{sec:applications}
While \eqref{eq:model} may look too simple to model signals of
interest for real-world applications, the following three scenarios
are of interest and fit this model. 

\subsubsection{X-ray crystallography}
The is the primary technique to determine the structure of
molecules. The experiment consists of the following steps: first, the
molecule $e(\vx)$ of interest is crystallized. The obtained crystal
$f(\vx)$ is simply a periodic repetition of the basic structure
$e(\vx)$, called \emph{unit cell},
\begin{align}
 f(\vx)=\sum_{\vi\in\Z^D}e(\vx-\vi\otimes\vN), \nonumber
\end{align} where $\vN$ is a vector containing the sizes of the unit
cell in each dimension. Second, the crystal is exposed to an X-ray
beam, under different angles.  For each angle, we have a diffraction
pattern that, mathematically speaking, is a slice of the
three--dimensional FT of the crystal.  See Figure \ref{fig:crystal}
for a graphical depiction of a unit cell, a crystal and a diffraction
pattern. The diffraction patterns are recorded using traditional
imaging techniques, such as CCDs, and only the magnitude is
acquired. Hence, we aim at recovering the spatial distribution
$e(\vx)$ of the molecule, called \emph{electron density}, from the
magnitude of the FT of $f(\vx)$. 

Due to the periodicity of the crystal, we are in fact  measuring the
magnitude of the Fourier series coefficients of the crystal,
$|\widehat{f}_{\vu}|$. This set of coefficients is equal, up to a
constant factor, to samples of the magnitude of the unit cell FT,
$|e(\vx)|$.

\begin{remark} The PR problem in X-ray crystallography is more complex
  than Problem \ref{prob:pr_c}. In fact, the set of measured samples
  is not sufficiently dense to reconstruct $|\widehat{e}(\vu)|^2$
  using Shannon's sampling theorem. More precisely, we have an
  undersampling factor of two for each dimension and we do not dispose
  of the entire ACF of $e(\vx)$. See  \cite{Millane:1990pt} for more
  details about this remark. 
\label{rmk:undersampling}
\end{remark} 

There is some a-priori information about the crystal that we can
exploit. For example, we can realistically model the unit cell
$e(\vx)$ as
\begin{align} e(\vx)=\sum_{n=1}^Nc^{(n)}\phi(\vx-\vx^{(n)}), \nonumber
\end{align} where $\phi(\vx)$ is the electron density of a single
atom\footnote{In reality, every atom has a different electron density
  $\phi(\vx)$. However, this assumption is reasonable and many
  reconstruction methods used in crystallography consider similar
  assumptions \cite{Woolfson:1963ec}.}  that has a positive coefficient
$c^{(n)}$ and is located at $\vx^{(n)}$. We can now specify the
PR problem for crystallography.
\begin{problem}
  \label{prob:crystallography} {\bf PR for Crystallography} Consider a
  unit cell $e(\vx)$ with $N$ \emph{positive} atoms on
  a bounded domain. Given a set of magnitudes of the Fourier series
  coefficients $|\widehat{f}_{\vu}|$, estimate the locations and
  amplitudes of the atoms.
\end{problem}
\subsubsection{Speckle imaging in astronomy}
Another example of the PR problem can be found in astronomy, namely an
imaging method known as \emph{speckle imaging}. This technique
attempts to mitigate the resolution downgrade introduced by
atmospheric turbulences. Namely, the atmosphere blurs $M$ images
$\{g^{(i)}(\vx)\}_{i=1} ^M$ collected by a telescope and the blurring
is modeled as a linear filter that may vary for each image,
$\{s^{(i)}(\vx)\}_{i=1}^M$. The $i$-th measured image, also called
speckle, is the convolution between the astronomic object $f(\vx)$ and
the $i$-th linear filter
\begin{align} 
g^{\,(i)}(\vx)=f(\vx)*s^{\,(i)}(\vx), \quad i=1,\dots,M. \nonumber
\end{align} 
See Figure \ref{fig:speckles} for an example of the speckles $g^{\,
  (i)}(\vx)$ and the target of the astronomic observations $f(\vx)$.

We reduce the atmospheric distortion and some potential additive white
noise by taking the average of the squared magnitudes of FT of the
images as
\begin{align} \frac{1}{K}\sum_{i=1}^{M}|{\widehat
g}^{\,(i)}(\vu)|^2=|\widehat{f}(\vu)|^2\frac{1}{M}\sum_{i=1}^M|\widehat
s^{\,(i)}(\vu)|^2, \nonumber
\end{align} 
where ${\widehat g}^{\,(i)}(\vu)$, $\widehat{f}(\vu)$ and $\widehat
s^{\,(i)}(\vu)$ are the FT of the measured image, the object of
interest, and the transfer function of the atmosphere,
respectively. Note that the averaging strategy is effective since we
assume that the atmospherical transfer functions generally affects
only the phases of $\widehat{f}(\vu)$ \cite{Wang:1978wl}. The averaged
atmospherical transfer function $\sum_{k=1}^M|\widehat s^{\,
  (i)}(\vu)|^2$ is estimated using atmospheric models or images of a
reference astronomical object.

\begin{figure*}[t!]
  \centering
  \subfloat{\includegraphics[scale=0.82]{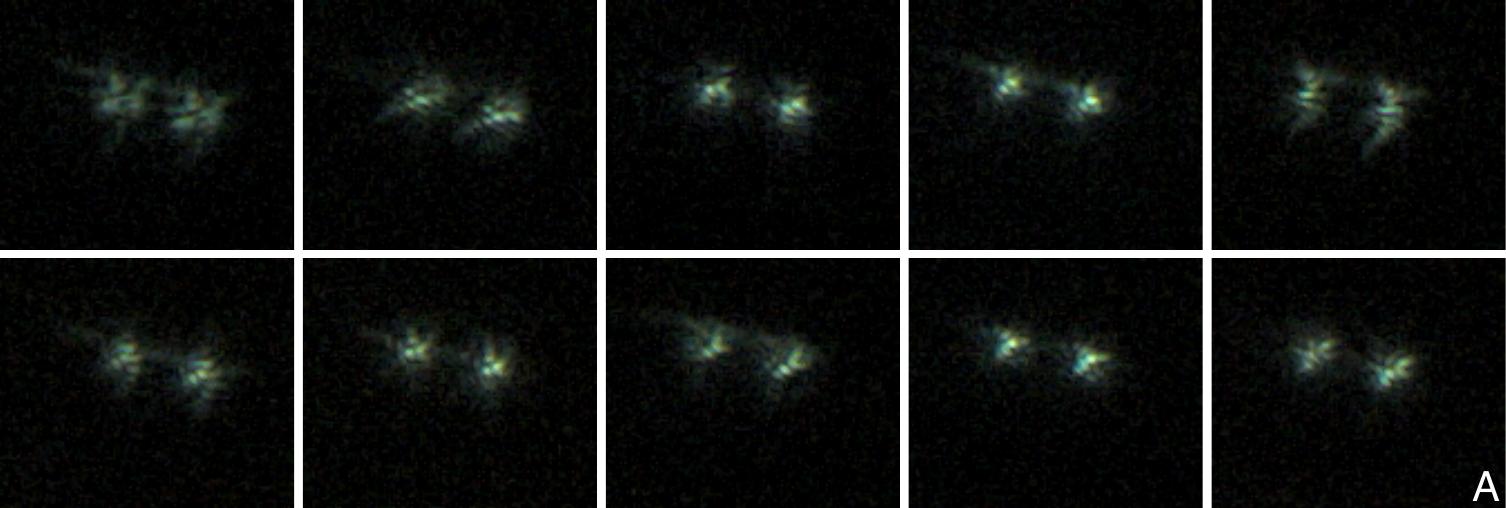}}
  \qquad
  \subfloat{\includegraphics[scale=0.82]{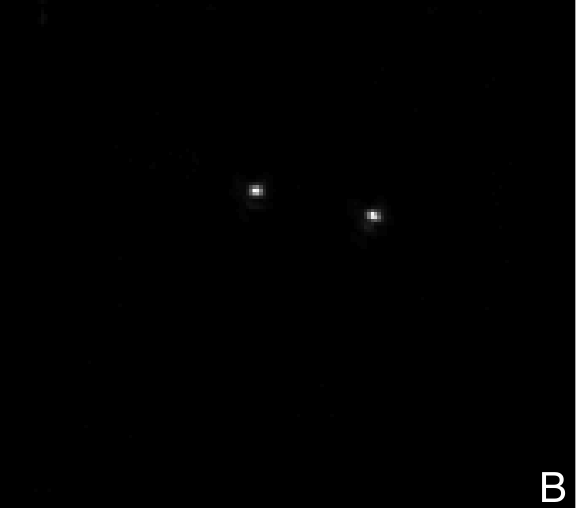}}       
  \caption[]{ An example of the input and output data of Problem
  \ref{prob:imaging}. In (A) we can see a set of 10 speckle images
  of a double star called $\epsilon$ Lyrae collected by Josef Popsel
  at the Capella Observatory, Mount Skinakas, Crete, Greece
  \cite{Hirsch:2011bx}. In (B), the high resolution image of
  the stars obtained through PR.}\vspace{-5mm}
  \label{fig:speckles}
\end{figure*}

PR is necessary to recover the high resolution image of the astronomic
object of interest from $|\widehat{f}(\vu)|^2$. Note that we
introduced the problem considering a continuous model for both the
astronomical object and the measured images. However, since the images
are generally measured and processed as sampled data, we may consider
a set of discrete images $\{g^{(i)}_{\vk}\}_{i=1}^M$ and the DFT is
usually employed.

We assume that $f(\vx)$ is sparse, 
\begin{align}
f(x)=\sum_{n=1}^N c^{(n)}\delta(\vx-\vx^{(n)}), \nonumber
\end{align}
since the astronomical object is composed of a set of stars that can
be modeled as Dirac deltas. Therefore, we have the following PR
problem for speckle imaging.

\begin{problem}
  \label{prob:imaging} {\bf PR for Speckle Imaging} Let $f(\vx)$ be a
  $N$-sparse function and $\{g^{(i)}(\vx)\}_{i=1}^M$ be a set of
  speckle images distorted by the atmosphere. Assume we measure a set
  of sampled images $\{g^{(i)}_{\vk}\}_{i=1}^M$, where $\vk\in\Z^D,\;
  D\ge 2$. Estimate the locations $\vx^{(n)}$ and the coefficients
  $c^{(n)}$ of the stars.
\end{problem}

\subsubsection{Blind channel estimation}
We conclude this section with an interesting example of 1--dimensional
PR: blind estimation of a communication channel. The knowledge of the
channel impulse response is fundamental for wireless communications
systems, such as the ones based on Orthogonal Frequency Division
Multiplexing \cite{Nee:2000:OWM:555664}. According to the theory of
multi-path propagation \cite{Kang:1999vi}, the channel $g(t)$ can be
faithfully modeled as 
\begin{align} 
g(t)=\sum_{n=1}^Nc^{(n)}\delta(t-t^{(n)}),
\label{eq:channel}
\end{align} where $t$ is the time variable and the $N$ deltas describe
the multi-path phenomenon. More precisely, each delta represents a
secondary communication path generated by a reflective body between
the source and the receiver.

We would like to estimate the locations and the coefficients of the
deltas, without having direct control of the channel input, hence the
``blind estimation'' terminology. We only measure samples of the
output $y(t)$ of the channel, that is the convolution between the
input $f(t)$ and the channel itself,
\begin{align} y(t)=(g*f)(t).\nonumber
\end{align}

The input data $f(t)$, which is the result of the modulation of a
discrete sequence $x_n$, is usually whitened to achieve the maximum
capacity of the channel. If the input sequence is statistically white,
then the magnitude of the FT of the output signal is in expectation
equal to the magnitude of the FT of the channel. Once more, we cannot
access the continuous-time output $y(t)$, but only a set of $M$
samples $\{y_{k}\}_{k=1}^M$. To recover the channel $g(t)$, we can
take the DFT of the collected samples, keep the magnitudes and solve
the following PR problem.

\begin{problem}
  \label{prob:channel} {\bf PR for Blind Channel Estimation} Let
  $g(t)$ be a multi-path fading communication channel as defined in
  \eqref{eq:channel}, where $N$ is finite and generally small. Assume
  the input of the channel to be properly whitened. Then, estimate the
  channel impulse response $g(t)$ from a set of samples
  $\{y_{k}\}_{k=1}^M$ of the output $y(t)$.
\end{problem}

We conclude this list of applications emphasizing the \emph{leitmotif} connecting all the
different applications: we are interested in PR for $N$--sparse signal
$f(\vx)$ defined on a $D$--dimensional continuous domain. We collect
samples of the signal and the phase information is lost, as in blind
channel estimation, or irreparably distorted, as in speckle
imaging. We would like to recover the sparse components in the
continuous domain, without discretizing the solution's domain.

This approach already proved beneficial in other domains. For example,
it has been shown that it is possible to recover a $N$-sparse signal
$f(x)$ from only $2N+1$ samples of the filtered signal $(g*f)(x)$, see
\cite{Vetterli:2002bs}. Another example where the continuous-time model
has been proven to be effective is in channel estimation. In
\cite{Barbotin:2012vt}, the authors demonstrated that the channel
estimator based on the continuous-time model achieves better
performance when compared to the state-of-the-art discrete approaches.

\section{Literature review}
\label{sec:literature}
We present a literature review that covers theoretical and algorithmic
results for both continuous and discrete PR. Note that most of the
works focused on the latter, given the difficulties of treating the PR
for continuous signals.

\subsection{Continuous PR}
Most of the relevant works connected to the continuous sparse PR
problem was developed in combinatorics for the \emph{turnpike} problem
\cite{Lemke:2002um}. The turnpike problem deals with the recovery of
the locations of a set of points from their unlabeled distances. Note
that the recovery of the support of $f(\vx)$ from the support of
$a(\vx)$ is an instance of such problem. A theorem presented by
Piccard in 1939 \cite{Piccard:1939tr} gives a sufficient condition for
the uniqueness of the turnpike problem. Unfortunately, a
counterexample to the theorem was first found by Bloom et
al. \cite{Bloom:1975vq} and its generalization was recently obtained
by Bekir et al. \cite{Bekir:2007kx}. A similar but weaker condition
for multidimensional signals has been recently obtained by Senechal in
2008 \cite{Senechal:2008ds}. Skiena et al. \cite{Lemke:2002um}
proposed a non-trivial algorithm for solving the problem. It is known
as the \emph{backtracking algorithm} and solves any instance of the
turnpike problem providing the existence of (possibly multiple) valid solutions.
The algorithm has a polynomial computational complexity when the set
$\{\vx^{(n)}\}$ is drawn at random. Zhang \cite{Zhang:1994td} showed
how to build sets of points achieving the worst case computational
complexity, that is $\mathcal{O}(2^n n\log n)$.

An equivalent problem has been stated for \emph{restriction site
  mapping}, an interesting task in computational molecular biology,
where a particular enzyme is added to a DNA sample, so that the DNA is
cut at particular locations $\{x^{(n)}\}_{n=1}^N$, known as
restriction sites. One can find the distance between each pair of
restriction sites, $\{x^{(n)}-x^{(m)}\}_{n,m=1}^N$, using gel
electrophoresis. Given the distances, we would like to recover the
locations of the sites. This technique is used for DNA mapping and it
usually involves different enzymes. When a single enzyme is used, it
is known as \emph{partial digest} \cite{Skiena:1994uo}. Note that it
has been shown by Cieliebak et al.  \cite{Cieliebak:2003uq} that the
partial digest problem with noisy measurements is $\mathsf{NP}$-hard.
The partial digest problem is in fact a turnpike problem with integer
locations---a bridge between continuous and discrete PR problems.

\subsection{Discrete PR}
Most of the literature focused on the discrete PR, that is when
$\vx\in\Z^D$, aiming to reduce the complexity of the solution.  The
first studies of the discrete PR problem appeared in control theory
and signal processing, where PR has been studied for the estimation of
the Wiener filter. In these fields, PR is known as spectral
factorization and a review of its theory and of the related algorithms
is given in \cite{Sayed:2001xr}. Among the presented methods, the
so-called ``Bauer'' method, described in \cite{Bauer:1955qo}, is the
most interesting one given the performance and the elegant matrix
formulation. However, the theory of spectral factorization is
focused on minimum phase solutions, that are stable and causal. These
solutions are not of interest for the applications given in Section
\ref{sec:applications}.

The uniqueness of the discrete PR problem has been studied for
multidimensional discrete signals. One of the main results is given by
Hayes \cite{Hayes:1982ud}: the set of positive finitely supported
images $f_{\vk}$ which are not uniquely recoverable has measure
zero. The results are derived using the theory of multidimensional
polynomials. A possible algorithm to recover signals from the
magnitudes of the FT is also given, but it does not achieve
satisfactory results according to the authors. Note that this
uniqueness result cannot be directly applied to X-ray crystallography
(see Remark \ref{rmk:undersampling}).

On the algorithmic front, many reconstruction algorithms were developed
for Problem \ref{prob:crystallography} and a review is given in
\cite{Sayre:2002if}. Among them, \emph{ab-initio} or \emph{direct}
methods were introduced in the late 50s and have the considerable
advantage of not requiring any prior information regarding the
crystals. The state of the art among these type of algorithms is
\emph{charge flipping} \cite{Oszlanyi:2004gb}. It performs two
operations iteratively, one in the spatial domain where it imposes the
positiveness of the electron density and the bounded support, one in
the Fourier domain where it imposes the measured magnitudes. This
algorithm was first presented in 2004 \cite{Oszlanyi:2004gb}, while
some of the recent developments are described in
\cite{Oszlanyi:2008fx}. It can be seen as a version of the
Gerchberg-Saxton algorithm \cite{Gerchberg:1972un}, where the
positivity constraint is enforced if the electron density is above a
certain threshold or set to zero otherwise.

Recent work defined efficient convex relaxations for solving discrete
PR problems. These approaches are potentially more stable w.r.t noise
and are capable of avoiding local minima. This strategy has been
introduced by Lu et al. \cite{Lu:2011uz}, who described a necessary
condition for the uniqueness of PR for discrete 1--dimensional sparse
signals together with a reconstruction algorithm based on the lifting
of the problem in a higher dimensional linear space that is solved by
traditional convex optimization solvers. A similar approach, but
introducing random masks to improve the redundancy of measurements,
has been introduced by Cand{\` e}s et al.,
\cite{Candes:2011ve,Candes:2011uz}. Hassibi and his collaborators
\cite{Jaganathan:2012ta} proposed an improved algorithm and sufficient
probabilistic uniqueness conditions based on the sparsity of the
signal of interest. Waldspurger et al. \cite{Waldspurger:2012wu}
formulated another tractable convex relaxation similar to the
classical MaxCut semidefinite program \cite{Goemans:1995ug} that
achieves better reconstruction performance when compared to other
convex relaxations.

PR has been generalized to any linear operator beyond the FT. More
precisely, we choose a frame $\set{\calF_i}_{i=1}^{K}$ and we collect
measurements of an unknown signal $f$ using the elements of this frame
$\set{w_i}_{i=1}^K=\set{\left<\calF_i, f
  \right>}_{i=1}^K$. Equivalently to the PR problem, we assume we can
only rely on the magnitude of the measurements and obtain a
generalized PR problem: from the magnitude of the expansion
$\set{|w_i|}_{i=1}^K$, recover the original signal $f$. Note that if
the chosen frame is the Fourier frame, then we have an instance of PR
on a discrete domain.  Balan et al. formulated the problem and studied
different theoretical and algorithmic aspect of the recovery of
signals from the magnitude of generic frame
coefficients. Specifically, fast algorithms are given in
\cite{R.Balan:2007tw}, while the statement of equivalent problems and
the construction of particular frames for which the reconstruction is
unique are given in \cite{Balan:2007jw,Balan:2005mb},
respectively. Relevant uniqueness results are given by Chebira et
al. \cite{Chebira:2010os}, where a necessary and sufficient condition
for uniqueness has been described, however it requires exponential
time to be checked. Note that the aforementioned convex relaxations
can be applied to this generalized PR.

Other generalization of the PR problem have been considered. Oppenheim
et al. proved the uniqueness of PR problems up to the knowledge of the
signs of the Fourier coefficients in \cite{Van-Hove:1983yo}. A general
analysis of the phase loss and the magnitude loss is given in
\cite{Hayes:1980ss} and an extension to the multidimensional case is
given in \cite{Hayes:1982ud}. The main result concerns the uniqueness
and the reconstruction of the magnitude loss problem.


\section{Uniqueness of the sparse PR problem}
\label{sec:uniqueness} In this section, we state sufficient conditions
to have a unique PR for sparse signals, whether discrete or continuous. We use
the same notation for both problems,
$f(\vx)=\sum_{n=1}^Nc^{(n)}\delta(\vx-\vx^{(n)})$, where $\vx^{(n)}$
is constrained to the set of integers $\Z^{D}$ for the discrete
problem. 

We use a divide and conquer approach to derive the uniqueness
condition. First, we notice that the locations of the deltas of the
ACF contain more \emph{information} than their coefficients. In fact, if
all the deltas have the same coefficient, the coefficients of the ACF do
not carry any information. Therefore, we consider the problem of
recovering the support of $f(\vx)$ given the support of its ACF
$a(\vx)$.  We then use the coefficients to further restrict the
possibility of having a non-unique solution.

We consider the set of locations $\set{\vx^{(n)}}_{n=1}^N$ and derive
the \emph{set of differences}
$\mathcal{D}=\set{\vx^{(n)}-\vx^{(m)}}_{n,m=1}^ {N,N}$. Given the
possibility of repeated elements, $\mathcal{D}$ is formally a
multiset. Looking at \eqref{eq:ACF_sparse}, we also notice that all
the elements of a sparse ACF are supported on the set of differences,
that is $\vy^{(n)}\in\cal{D}$.

If we attempt to recover the support of a sparse signal from the
support of its ACF, we realize that the lack of labeling of the
elements of $\cal{D}$ makes the problem combinatorial, that is all the
possible labelings must be tested to find the optimal
solution. Moreover, the solution may be even more complex if the ACF
has ``\emph{collisions}'': two deltas of the ACF located at the same
position due to two couples of equi-spaced deltas in the signal
$f(\vx)$.

\begin{definition}[Collision] We say there is a collision in the ACF
  when $\exists\, n\neq m \text{ such that } \vy^{(n)}=\vy^{(m)}$.
\end{definition} 

In other words, let $\vx^{(i)}$,$\vx^{(j)}$,$\vx^{(k)}$ and
$\vx^{(l)}$ be the locations of four distinct deltas of a sparse
signal $f(\vx)$, then we say that we have a collision in the ACF if
$\vx^{(i)}-\vx^{(j)}=\vx^{(k)}-\vx^{(l)}$.  

The main reason why collisions are problematic is the impossibility of
knowing a-priori how many $\vy^{(n)}$ are colliding on the same
element of the ACF that we observe. In what follows, we show that if
the observed ACF $a(\vx)$ does not have collisions, we are able to
recover uniquely the sparse signal $f(\vx)$ in most of the cases.

\subsection{Uniqueness condition: collision-free 1--dimensional ACFs}
\label{sec:1d}

We assume that we first want to recover the support of $f(\vx)$ from
the support of $a(\vx)$. Here we study the problem for a
one--dimensional signal $f(x)$. We consider the set of differences
$\mathcal{D}$ as the input and the locations of the deltas
$\set{x^{(n)}}_{n=1}^{N}$ as the output of the following problem:
\begin{problem} {\bf Support Recovery} Given all the pairwise
  distances $\{x^{(n)}-x^{(m)}\}_{n,m=1}^{N,N}$ between a set of $N$
  points lying on a 1--dimensional domain, recover their locations
  $\{x^{(n)}\}_{n=1}^N$.
\label{prob:rec_points}
\end{problem} Note that we have no information about the labeling of
the pairwise differences in $\cal{D}$. Therefore, Problem
\ref{prob:rec_points} is combinatorial and is equivalent to an
instance of the \emph{turnpike problem} \cite{Lemke:2002um}.  If the
naming was known, the problem could be easily solved by
\emph{multidimensional scaling} \cite{Cox:2000uq}.

We introduce the definition of homometric sets to define the uniqueness
of Problem \ref{prob:rec_points}.
\begin{definition} [Homometric Sets] Two sets $\calX$ and $\calY$
are said to be homometric if and only if their difference sets are
congruent, that is $\mathcal{D}_\calX=\mathcal{D}_\calY$.
\label{def:homometric}
\end{definition}
\noindent In this section, we assume that all the differences are
different from each other, i.e. we do not have any collision in the
ACF. This is equivalent to saying that the set $\mathcal{D}$ has no
repeated elements. It turns out that the problem of the support
recovery was first posed by Patterson \cite{Patterson:1935wy} and a
possible solution was proposed by Piccard in 1939
\cite{Piccard:1939tr}. More precisely, Piccard suggested that if there
are no collisions, the solution of the turnpike problem is always
unique.

Unfortunately, a counterexample to this result was found
in 1975 by Bloom \cite{Bloom:1975vq}. Consider
$\mathcal{X}=\set{0,1,4,10,12,17}$ and
$\mathcal{Y}=\set{0,1,8,11,13,17}$, then
\begin{align}
\mathcal{D}_\mathcal{X}=\mathcal{D}_\mathcal{Y}=\set{0,1,2,3,4,5,6,7,8,9,10,11,12,13,16,17}.\nonumber
\end{align}
Recently, this counterexample has been proved to belong to a unique
parametric family of counterexamples by Bekir \cite{Bekir:2007kx}.

\begin{theorem} {\bf \cite{Bekir:2007kx}} 
  If $\mathcal{X}$ and $\mathcal{Y}$ are finite sets of points whose
  differences sets, $\mathcal{D}_\mathcal{X}$ and
  $\mathcal{D}_\mathcal{Y}$, contain no repeated elements, then the
  turnpike problem has always a unique solution unless the elements of
  $\mathcal{X}$ and $\mathcal{Y}$ belong to single and unique infinite
  parametric family of six elements. More precisely, given
  $\vp=(p_1,p_2)\in\R^2$ and the two following parametric sets,
  $\mathcal{X}=\set{0,p_1,p_2-2p_1,2p_2-2p_1,2p_2,3p_2-p_1}$ and
  $\mathcal{Y}=\set{0,p_1,2p_1+p_2,p_1+2p_2,2p_2-p_1,3p_2-p_1}$, the
  two difference sets are congruent,
  $\mathcal{D}_{\calX}=\mathcal{D}_{\calY}.$
\label{thm:ACF_support}
\end{theorem}

Note that the equations defining the elements of sets $\mathcal{X}$
and $\mathcal{Y}$ are linear combinations of $\vp$. This suggests that
we can geometrically characterize (with linear subspaces) the supports
of signals that generate a turnpike problem without a unique solution.

\begin{corollary}
  The sets of points that generate a turnpike problem without a unique
  solution belong to the following 2--dimensional linear subspaces,
\begin{align}
\begin{bmatrix} 0& 0\\ 1& 0\\ -2 & 1\\ -2 & 2\\ 0 & 2\\ -1 & 3
\end{bmatrix}
\begin{bmatrix} p_1\\ p_2
\end{bmatrix}=\mQ_{\calX}\vp, \quad \begin{bmatrix} 0& 0\\ 1& 0\\ 2 & 1\\ 1 & 2\\ -1 & 2\\ -1 & 3
\end{bmatrix}
\begin{bmatrix} p_1\\ p_2
\end{bmatrix}=\mQ_{\calY}\vp.\nonumber
\end{align}
\label{cor:sig_support}
\end{corollary}

Even if these two linear subspaces define completely the sets
$\cal{X}$ and $\cal{Y}$, we need to define a proper ordering of the
elements. In fact, while the sets do not have by definition a defined
order, the two linear subspaces induce a precise ordering. We can
choose any permutation as soon as it is univocally defined: in what
follows, we always consider the permutations $\mat{\Pi}_{{\calY},\vp}$
and $\mat{\Pi}_{{\calX},\vp}$ that sort in an increasing order the
vectors $\mat{\Pi}_{{\calY},\vp}\mQ_{\calY}\vp$ and
$\mat{\Pi}_{{\calX},\vp}\mQ_{\calX}\vp$, respectively. Equivalently,
we always consider an operator that takes the elements from the sets
$\cal{X}$ and $\cal{Y}$ and sorts them in an increasing order.

Note that the permutation and the operator are unique for each $\vp$
and depend only on the direction of $\vp$: in fact, for different
magnitudes of $\vp$, the ordering does not change.  Moreover, it is
easy to show that we have a finite number of permutations for all the
$\vp$. In fact, the number of permutations is upper bounded by 6
factorial, being the number of possible permutation of 6 elements in a
set.

We defined the geometry of these sets of points as a manifold
generated by a linear model and a varying permutation. In what
follows, we take the span of these permuted linear systems to describe
the sets of supports without a unique solution to the turnpike
problem.

\begin{figure}[t!]  \centering
\includegraphics[scale=0.85]{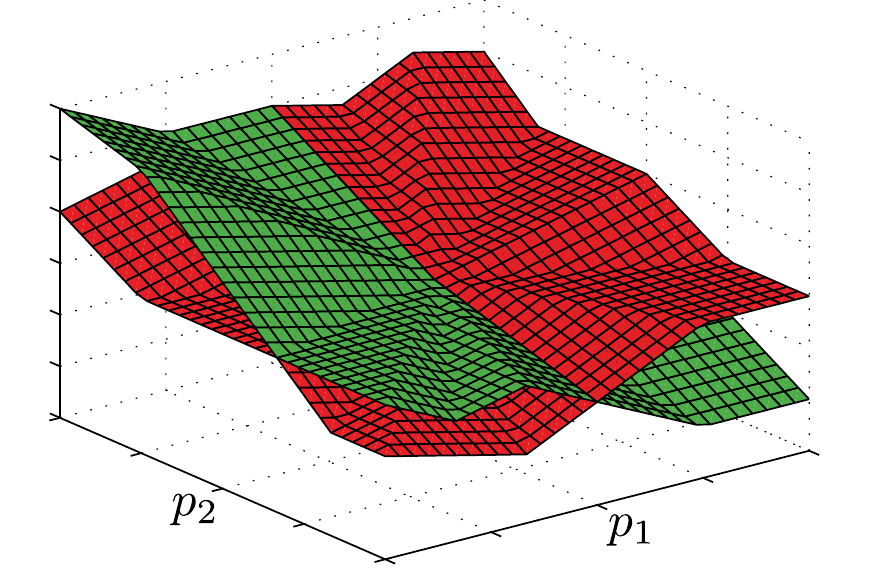}
\caption{ A low--dimensional projection of the two 6--dimensional
  spaces representing the supports $\calX$ (red) and $\calY$
  (green). Note that both are formed by the union of many subsets of
  2--dimensional subspaces, one for each permutation. Moreover, the
  intersection between the two sets has measure zero.}\vspace{-5mm}
  \label{fig:perm_set}
\end{figure}
\begin{corollary}\label{cor:perms}
  The set of supports for a signal $f(x)$ that generate a turnpike
  problem without a unique solution is described as,
\begin{align}
\mathcal{S}=\calX\cup\calY=\operatorname{span}(\mat{\Pi}_{{\calY},\vp}\mQ_{\calY}) \cup
\operatorname{span}(\mat{\Pi}_{{\calX},\vp}\mQ_{\calX}). \nonumber
\end{align}
\end{corollary}
A graphical representation of this set of elements is given in Figure
\ref{fig:perm_set}, where we observe the subset of linear subspaces
\emph{mixed} by the permutations.

This geometrical intuition is useful for two reasons:
\begin{itemize}
\item if we have no collisions and unless $N=6$, we can always recover
  uniquely the support of the signal from the support of the ACF,
\item even if $N=6$, the supports without a unique recovery lie on the
  2--dimensional manifold defined in Corollary \ref{cor:perms}, and
  this manifold has measure zero in the set of all the supports of 6
  elements.
\end{itemize}
Note that we have not used the coefficient information so far. The
following theorem merges the previous results in terms of phase
retrieval for sparse signals and considers the coefficients of the ACF
to obtain a sufficient condition for the uniqueness of 1--dimensional
sparse PR problems.

\begin{theorem} Assume we measure the 1--dimensional ACF of a signal
  with $N$ deltas and the elements of the ACF have no
  collisions. Then,
\begin{itemize}
\item If $N\neq 6$, the PR problem has a unique solution.
\item If $N=6$ and not all the $c^{(n)}$ have the same value, the PR
problem has a unique solution.
\item If $N=6$ and all the $c^{(n)}$ have the same value, the PR
problem has almost surely a unique solution.
\end{itemize}
\label{thm:unique_1}
\end{theorem}\noindent
The proof is given in Appendix \ref{app:1}. The three cases are due to
the parametric family of turnpike problems without a unique solution
previously described and the additional information that may be
available from the coefficients $c^{(n)}$ of the deltas of the ACF.

\subsection{Uniqueness condition: collision-free $D$--dimensional ACFs}
\label{sec:multi}
The previous analysis applies only to 1--dimensional signals, that is
$D=1$. Senechal et al. \cite{Senechal:2008ds} proposed an analysis of
the uniqueness of the turnpike problem in higher dimensions, $D\ge
2$. Unfortunately, their result is too conservative and cannot cover
very simple examples such as 1D sets of points embedded in higher
dimensional spaces.

In what follows, we describe the result of Senechal et al. and propose
a sufficient condition for the uniqueness of the PR for
multi--dimensional sparse signals.

\begin{definition} [Visible Point] A point belonging to a multi--dimensional
  difference set, or equivalently a delta of the ACF, is
  \emph{visible} if the line through the origin and the point contains
  only the origin, the point itself and another point in the
  centro--symmetric position w.r.t. the origin.
\end{definition}

Note that there is always a point in a centro--symmetric position, see
\eqref{eq:ACF_sparse_L}. In Figure \ref{fig:vis_points}, we show an
ACF with some visible deltas, such as $\vy^{(j)}$, and some deltas
that are not visible, such as $\vy^{(n)}$.
\begin{figure}[t!]  \centering
\includegraphics[scale=0.85]{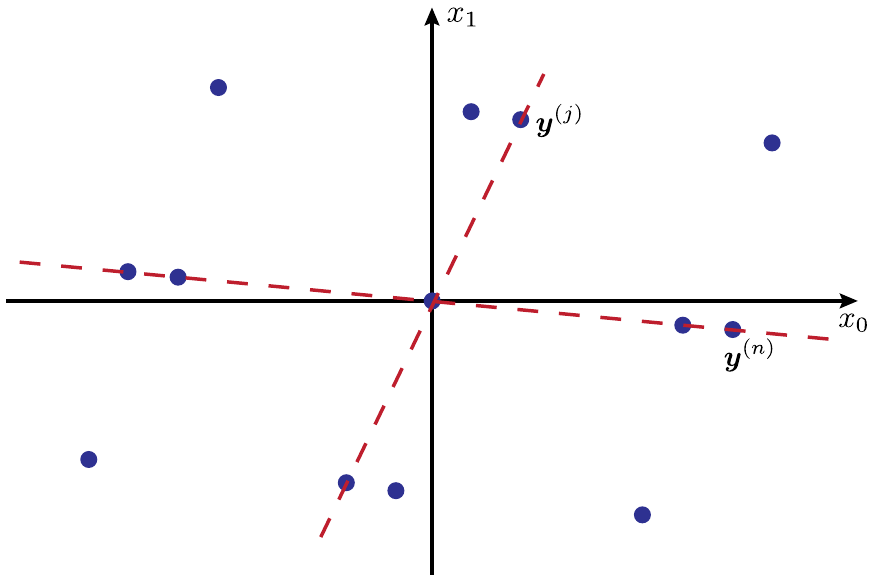} \vspace{-3mm}
\caption{A graphical representation of a 2--dimensional ACF generated
  by a $4$--sparse signal, where the blue dots represent the location
  of the deltas. Note that the delta located at $\vy^{(j)}$
  is visible while the one located at $\vy^{(n)}$ is not, since there
  is another delta aligned on the line passing through the origin.  }\vspace{-5mm}
  \label{fig:vis_points}
\end{figure}

\begin{definition} [General Position] A signal $f(\vx)$ composed of $N$ deltas has its
elements in \emph{general position} if every delta of the ACF is
visible.
\end{definition} Senechal et al. \cite{Senechal:2008ds} showed that
having the points in general position is a sufficient condition to
recover uniquely the support of a sparse signal $f(\vx)$ from the
support of its ACF.
\begin{theorem} {\bf \cite{Senechal:2008ds}} Let the signal $f(\vx)$ be
  supported on $\set{\vx^{(n)}}_{n=1}^N$ in $\R^D$ and let the deltas
  of its ACF be in general position. Then, we can uniquely recover the
  support of $f(\vx)$ from its ACF.
\label{thm:unique_2}
\end{theorem}
For example, Theorem \ref{thm:unique_2} does not guarantee a unique
recovery of the signal generating the ACF given in Figure
\ref{fig:vis_points}.  Note that while the proposed uniqueness
condition for the one--dimensional PR problem described in Theorem
\ref{thm:unique_1} and the one for multi--dimensional ones given in
Theorem \ref{thm:unique_2} are both sufficient, the latter is fairly
constraining. For example, if we pick a 1--dimensional sets of points
that generates a turnpike problem with a unique solution and embed it
into a higher dimensional domain, then Theorem \ref{thm:unique_2}
cannot guarantee anymore the uniqueness of the solution.

This example inspired the idea of solving a $D$--dimensional PR as a
set of multiple $1$--dimensional PR problems. In fact, if we can solve
uniquely the sub--problems, then the original problem may have a
unique solution.  In this section, we show that this divide and
conquer strategy leads to a tighter necessary and sufficient
condition: \emph{given an ACF $a(\vx)$ of a $N$--sparse
  $D$--dimensional signal $f(\vx)$, the PR problem has a unique
  solution}.

We show the result as follows: first we consider a set of projections
of the ACF to many 1--dimensional subspaces\footnote{It is possible to
  consider projections onto subspaces of any dimensionality. Here, we
  consider only 1--dimensional projections for simplicity of
  notation.}.  See Figure \ref{fig:proj} for an example of a projection
over a subspace defined by a vector $\mP$. Second, we show that the
projection of the ACF is the ACF of the projected signal. Third, we
show that a finite number of different projections is necessary and
sufficient to recover the deltas of $f(\vx)$. We conclude showing that
we can find these projections for every $N$-sparse $f(\vx)$ that is
embedded in a $D$--dimensional space, with $D\ge2$.

\begin{figure}[t!]  \centering
\includegraphics[scale=0.85]{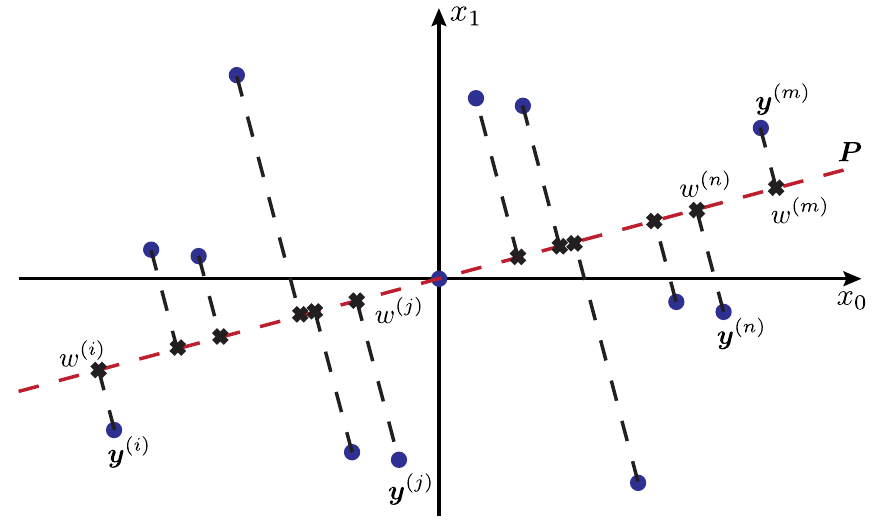} \vspace{-3mm}
\caption{The projection of a 2--dimensional ACF onto a 1--dimensional
subspace defined by a vector $\mP$. Note that the main properties of
an ACF, such as the centro--symmetry, are preserved in the projected domain.}\vspace{-5mm}
  \label{fig:proj}
\end{figure}

\begin{proposition} 
  Let $f(\vx)$ be a signal composed of $N$ deltas lying on $\R^D$ and
  let $a(\vx)$ be its ACF. Define a projection over a 1--dimensional
  space as $\mP\vx$, where $\mP$ is a $1\times D$ vector over which we
  project the spatial domain and $\vx$ are the coordinates in the
  original domain. Then the projected ACF is the ACF of the projected
  signal.
\label{prop:projection}
\end{proposition}
\noindent The proof is given in Appendix \ref{app:2}.

Given Proposition \ref{prop:projection}, we can project a
$D$--dimensional ACF $a(\vx)$ onto $P$ different 1--dimensional
subspaces. If the projected ACF satisfies the conditions given in
Theorem \ref{thm:unique_1}, we recover the projected signal, namely
the deltas with the proper coefficient and the projected locations. As
shown in \cite{Shukla:2006em}, we need $(N+1)^{(D-1)}$ projections to
reconstruct exactly the location of $N$ deltas in a $D$--dimensional
space. It is possible to reduce the number of required projection to
$D+1$ accepting to take random projections and exactly recovering the
support \emph{with probability one}.

Finally, we show the existence of the projections with a unique PR for
any $D$--dimensional ACF with $D\ge2$. More precisely, a projection
may unfortunately belong to the parametric family described in Theorem
\ref{thm:ACF_support} or it can have collisions in the projected
points. In what follows, we show that the projections without
collisions and with a unique PR exist and are easy to find for all the
$D$--dimensional signals with $D\ge 2$.  Indeed, it is simple to show
that for any ACF without collisions, the projected ACF on a random
subspace has no collisions with probability one. In the following
theorem, we show that the set of projections with a unique PR does
not have measure zero in the set of all projections for any ACF with
$D\ge 2$. In other words, the multi--dimensional sparse PR problem
always has a unique solution.

\begin{theorem}
  Let $f(\vx)$ be a $D$--dimensional $N$--sparse signal and let
  $a(\vx)$ be its ACF. Assume that $a(\vx)$ has no collisions and
  $D\ge2$. Then the set of projections onto 1--dimensional domains
  that generates a 1--dimensional PR problem with a unique solution
  has a measure larger than zero in the set of all the possible
  projections. Therefore, the solution to the PR of $f(x)$ is always
  unique.
\label{thm:md}
\end{theorem}
\noindent The proof is given in Appendix \ref{app:3}.

We conclude this section underlying that the uniqueness of the
solution of the sparse PR problem does not depend on the type of
domain---that is continuous or discrete. The presence of collisions is
a better characterizing feature.  However, the amount of collisions
depends on the domain: if we randomly distribute the deltas of
$f(\vx)$ on a continuous domain, we have almost surely no
collision, while the probability of collisions for discrete signals is
always larger than zero for $N\ge4$.

\section{Conclusions}

As we described in Section \ref{sec:applications}, many real-world
problems require the recovery a sparse signal on a continuous domain
from the magnitude of its Fourier transform. In this work, we provided
an answer to the uniqueness of the solution of such problem, whether
defined on a continuous or a discrete domain.

In particular, we showed that the uniqueness of the solution of such
problem depends on the presence of collisions. More precisely, the
solution is always unique for signals that are embedded on 2 or more
dimensions and whose ACFs do not have collisions. For 1--dimensional
signals, we showed the uniqueness for most of the signals with a
collision-free ACF, the only counterexample being signals composed of
$N=6$ elements with the same coefficient supported on a set of points
belonging to the unique family of counterexamples \cite{Bekir:2007kx}.

Future work will be focused on developing algorithms to solve the
1--dimensional PR problem on the continuous domain and its extension
to the multi--dimensional setup using multiple projections on
different 1--dimensional domains.

\appendix

\begin{figure*}[t!]  \centering
  \subfloat{\includegraphics[scale=0.85]{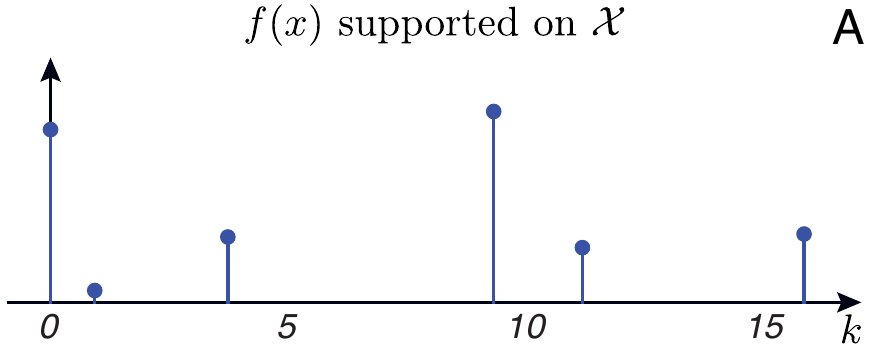}}
  \qquad
  \subfloat{\includegraphics[scale=0.85]{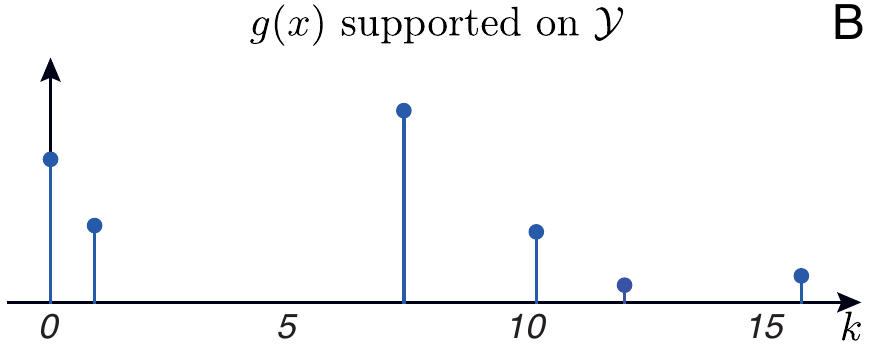}}\vspace{-0.1cm}\\ 
  \subfloat{\includegraphics[scale=0.85]{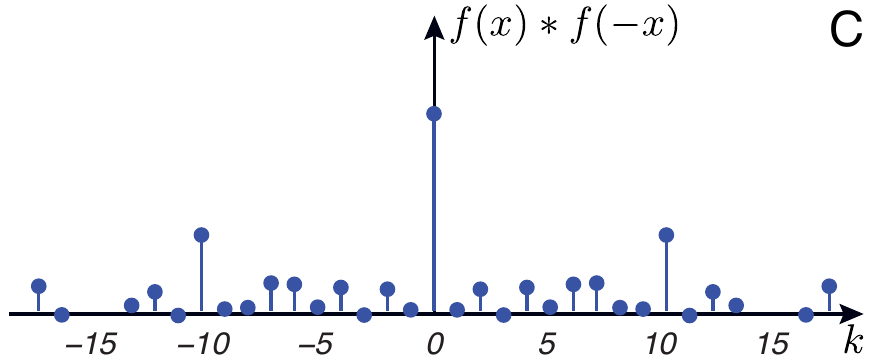}}
  \qquad
  \subfloat{\includegraphics[scale=0.85]{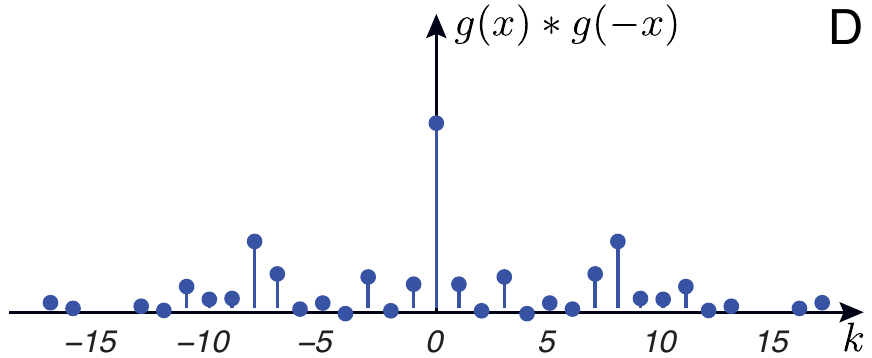}}\vspace{-2mm}
  \caption{ In (A) and (B), two random signals supported on $\cal{X}$
    and $\cal{Y}$ are shown. In (C) and (D), the respective ACF are
    given.  Note that even if the supports of the ACFs are same, the
    coefficients of the ACF depend significantly on the coefficients of
    the original signals. Here we considered as supports the first
    counterexample found by Bloom \cite{Bloom:1975vq} but an
    equivalent analysis can be done for the family of counterexamples
    described by Bekir et al. \cite{Bekir:2007kx}.}\vspace{-5mm}
  \label{fig:ex_domain}
\end{figure*}

\subsection{Proof of Theorem \ref{thm:unique_1}}
\label{app:1}
\begin{IEEEproof} We divide the proof in three parts, one for each
case given in the theorem's statement.

$\bullet$ For $N\neq 6$ we can always recover uniquely the locations of
the deltas from the sets of differences, see Theorem
\ref{thm:ACF_support}. Once the locations are known, the coefficients
are uniquely determined, see Appendix \ref{sec:ampl}. 

$\bullet$ If $N=6$ and all the $c^{(n)}$ have the same coefficient, and
the signal does not have a unique PR, then its support lies on the
manifold defined in Corollary \ref{cor:perms} and we cannot use the
coefficients to enforce uniqueness. Note that all possible signals just
$N=6$ span a 6--dimensional space, representing the six
locations. Given that the manifold containing the signals without a
unique reconstruction is 2--dimensional, then the set of these signals
has measure zero w.r.t. the set of all the signals with $N=6$.

$\bullet$ If $N=6$ and not all the $c^{(n)}$ have the same coefficient,
there is always only one set between $\mathcal{X}$ and $\mathcal{Y}$
that is a possible support of the signal $f(x)$.  To prove this
statement, we assume without loss of generality\footnote{If the deltas
  of the $f(x)$ are not positive, we can always take the absolute
  value given the absence of collisions by assumption. } that all
$d^{(n)}$ are positive and we would like to show that it is always
possible to discern the support of $f(x)$ between the two possibile
ones using the coefficients.

First, define two vectors $\vec{q}$ and $\vec{r}$ containing the
logarithms of the coefficients of the ACF and of the signal,
respectively. Then, given the absence of collision and the
linearization introduced by the logarithm we can write two systems of
equations. Each equation represents the coefficient of one element of
the ACF given the coefficients of two elements of the original signal,
see Figure \ref{fig:ex_domain}. The two systems assume that $f(x)$ is
supported on $\cal{X}$ and $\cal{Y}$, respectively. More precisely, if
we assume that $f(x)$ is supported on $\mathcal{X}$ we have
\begin{align}
\vec{q}=\boldsymbol{C}_\calX\mathbf{r},
\label{eq:syst1}
\end{align}
 while if we consider $\mathcal{Y}$ to be the support of
$f(x)$ we get a different system of equations,
\begin{align}
\vec{q}=\boldsymbol{C}_\calY\mathbf{r}.
\label{eq:syst2}
\end{align}

Successively, we note that if all the coefficients have the same value,
that is $c^{(n)}=C \; \forall\, n=\{1,\ldots,N\}$, all the ACF
elements have the same coefficient as well, $d^{(n)}=C^2$.  In this
case, \eqref{eq:syst1} and \eqref{eq:syst2} are equivalent and we
cannot distinguish between the two possible supports, $\calX$ and
$\calY$. Note that the set of signals having $c^{(n)}=C \; \forall n$
form a 1--dimensional subspace in the coefficient domain. In what
follows, we show that this is the only subspace where $\calX$ and
$\calY$ cannot be distinguished from the coefficients.

Consider the intersection between the two columns spaces of
$\boldsymbol{C}_\calX$ and $\mC_\calY$. This intersection contains all
the coefficients of the ACF that could be equivalently generated by two
signals, one supported on $\calX$ and one on $\calY$. The
dimensionality of this intersection can be computed as
\begin{align}\label{eq:boh}
\operatorname{dim}& (\operatorname{span}(\mC_\calX)\cap
\operatorname{span}(\mC_\calY))= \\\nonumber&
\rank(\boldsymbol{C}_\calX)+\rank(\boldsymbol{C}_\calY)-\rank([\boldsymbol{C}_\calX,\boldsymbol{C}_\calY])=1
\end{align} where $[\boldsymbol{C}_\calX,\boldsymbol{C}_\calY]$ is a
symbol representing the horizontal concatenation of the two matrices
$\boldsymbol{C}_\calX$ and $\boldsymbol{C}_\calY$. It is possible to
verify that \eqref{eq:boh} holds for the example of Bloom and by
linearity to any other element of the set of counterexamples.

We deduce from \eqref{eq:boh} that the ACFs with coefficients
$\set{d^{(n)}}_{n=0}^{N^2-N}$ that can be generated by the two
different supports lie on a 1--dimensional subspace. Given that we
have already characterized this 1--dimensional subspace as the one
where all the $\set{c^{(n)}}_{n=0}^5$ have the same value, this
concludes the proof.
\end{IEEEproof}

\subsection{Proof of Proposition \ref{prop:projection}}
\label{app:2}
\begin{IEEEproof} First, we define the projection $\bar{f}(s)$ of the
  signal $f(\vx)$ over the subspace indicated by $\mP$ as the
  following inner product
\begin{align} \bar{f}(s)=\int_{\R^D}\delta(s-\mP
\vx)f(\vx)d\vx=\sum_{n=1}^Nc^{(n)}\delta(s-\mP\vx^{(n)}).\nonumber
\end{align} Then, we compute the ACF of the projected signal as
\begin{align} &a(s) =\bar{f}(s)* \bar{f}(-s)\nonumber\\
  &=\int_{\R^D} \sum_{n=1}^Nc^{(n)}\delta(h-\mP\vx^{(n)})\sum_{m=1}^Nc^{(m)}\delta(h+s-\mP\vx^{(m)})dh\nonumber \\
  &=\sum_{n=0}^{N^2-N}d^{(n)}\delta(s-\mP\vy^{(n)}),
\label{eq:proj}
\end{align} where the $\delta(s)$ are located on the 1--dimensional
domain defined by the projection.  Finally, we compute the projection
of the ACF according to \eqref{eq:ACF_sparse} as
\begin{align} \bar{a}(s)
&=\int_{\R^D}\delta(s-\mP\vy)a(\vy)d\vy\nonumber \\
&=\int_{\R^D}\delta(s-\mP\vy)\sum_{n=0}^{N^2-N}d^{(n)}\delta(\vy-\vy^{(n)}) d\vy\nonumber\\ 
& =\sum_{n=0}^{N^2-N}d^{(n)}\delta(s-\mP\vy^{(n)}),
\end{align} that is equal to \eqref{eq:proj}, which proves the
proposition.
\end{IEEEproof}

\subsection{Proof of Theorem \ref{thm:md}}
\label{app:3}

\begin{IEEEproof} First, we introduce some notation. The projected
deltas are located on $\{w^{(n)}\}_{n=0}^{N^2-N}$. These locations can
be computed with the following matrix vector multiplications,
\begin{align}\mathbf{w}&=\begin{bmatrix} w^{(0)} \\ w^{(1)} \\ \vdots
    \\ w^{(N^2-N)} \end{bmatrix}\nonumber\\ 
&=\begin{bmatrix} y^{(0)}_0 &
    y^{(0)}_1 & \cdots & y^{(0)}_{D-1} \\ y^{(1)}_0 & y^{(1)}_1 &
    \cdots & y^{(1)}_{D-1} \\ \vdots & \vdots & & \vdots \\
    y^{(N^2-N)}_0 & y^{(N^2-N)}_1 & \cdots & y^{(N^2-N)}_{D-1} \\
\end{bmatrix}\begin{bmatrix} p_0 \\ p_1\\ \vdots \\ p_{D-1}
\end{bmatrix}\nonumber\\&=\mY\mP^{T},
\label{eq:thm_proj}
\end{align}

where $\mY$ is a matrix containing on each row the locations of the
points $\set{\vy^{(n)}}_{n=0}^{N^2-N}$.  Note that the set of all
possible projections from a $D$--dimensional to a $1$--dimensional
space forms a $\rank(\mY)$--dimensional space. More precisely, unless
the points are lying on a lower dimensional subspace, it forms a
$D$--dimensional subspace.

Similarly to the analysis proposed in Section \ref{sec:1d}, we can
geometrically characterize the supports of the ACFs for which we
cannot solve the turnpike problem uniquely.  We build a linear model
$\mQ$ from the difference set induced by $\cal{X}$, or $\cal{Y}$
equivalently. Then, we define a permutation matrix $\Pi_{\vp}$ to
match the ordering given by the linear model with the ordering chosen
for the set of points representing the ACF support.  Assume that the
projected points $\set{w^{(n)}}_{n=0}^{N^2-N}$ are supported on a set
of points without a unique PR. Then, we have
\begin{align}
\vec{w}=\mat{\Pi_{\vp}}\mat{Q}\vp,
\end{align}
where $\mQ$ spans a 2-dimensional linear subspace and when combined
with the $\vp$-dependent permutation $\Pi_{\vp}$, it defines a
2-dimensional manifold.

The permutation is fixed by $\vp$ and we have an intersection between
the $D$--dimensional subspace spanned by $\mY$ and the 2--dimensional
subspace spanned by $\mQ$. We analyze the possible cases w.r.t. the
dimensionality of the original ACF, that is given by $\rank{\mY}$:
\begin{itemize}
\item if the ACF is 1--dimensional and we have an intersection between
  $\operatorname{span}(\mY)$ and $\operatorname{span}(\mQ)$, this
  intersection has forcefully the size of all the possible
  projections $\rank(\mY\cap\mQ)=1$. Changing the direction of the
  projection has no effect beside a scaling. In this case, the original signal is
  not uniquely recoverable from the support, see Theorem
  \ref{thm:ACF_support}.
\item if the ACF is 2--dimensional and the intersection is
  1--dimensional, there is only one projection that does not have a
  unique reconstruction and we can pick any other $\mP$ to have a
  unique PR. On the other hand, if the intersection is 2--dimensional,
  then the signal subspace is the same as the one spanned by the
  matrix $\mQ$.  However, we can change $\vp$ such that the
  permutation $\mat{\Pi}_{\vp}$ changes and the projected points have
  a unique 1--dimensional PR.  Given the structure of the manifold
  induced by the permutations, the \emph{good} projections are
  relatively easy to find.
\item if the ACF is $D$--dimensional and $D>2$, the set of projections
  $\mP$ that generates projected PR problems with a unique solutions
  is dense in the set of the projections. In fact, the set of
  projections without a unique solutions is a 2--dimensional manifold
  while all the projections are a $D$--dimensional set. We always find
  sufficient projections to recover the support of $f(\vx)$ and
  therefore all the $f(\vx)$ without collisions have a unique PR.
\end{itemize} \end{IEEEproof}
\subsection{Recovering the coefficients from the support}
\label{sec:ampl}
In general, the problem of recovering the coefficients
$\{c^{(n)}\}_{n=1}^N$ of a $N$ sparse signal $f(x)$ knowing the
support $\{x^{(n)}\}_{n=1}^N$ and the ACF $a(x)$ is not trivial. It is
equivalent to solve a system of quadratic forms and a possible convex
relaxation is given in \cite{Jaganathan:2012ta}. Unfortunately, it is
not possible to guarantee the success of reconstruction. 

The problem has the same formulation whether the domain is discrete or
continuous. In what follows we derive an algorithm that recovers the
coefficients of the a sparse signal whose support and ACF are both
known while requiring the absence of collision in its ACF.

Let $\mathbf{c}=[c^{(1)},c^{(2)}, \ldots, c^{(N)}]^\top$ be a vector
containing the coefficients of a $N$-sparse signal $f(\vx)$ and define a
rank-one matrix $\mathbf{C}=\mathbf{cc}^*$

Given that we know the ACF and the support of $f(\vx)$, then we know all
the off-diagonal elements of the matrix $\mathbf{C}$. Therefore, we
can reformulate the task of recovering the coefficients from the support
as follows.

\begin{problem}
Consider a rank-one matrix $\mathbf{C}$ whose elements are all known
except the ones on the main diagonal.  Can we uniquely reconstruct the
elements on the main diagonal?
\label{prb:rank1}
\end{problem}

There exists many alternative approaches to solve Problem
\ref{prb:rank1}. In what follows, we describe a method that requires a
single matrix inversion. 

Let $\alpha_k\triangleq |c_k|^2$ $\forall k=1,\ldots,N$, and define a
matrix
$\mathbf{A}=\mathbf{C}-\operatorname{diag}(\alpha_1,\ldots,\alpha_N)$. Note
that $\mathbf{A}$ is known and we aim to recover the matrix
$\mathbf{C}$. The following result describes a method that requires a
single matrix inversion to find the unique solution of Problem \ref{prb:rank1}.

\begin{proposition}
Assume that $N>2$, then $$\mathbf{C}_{k,k}=|c_k|^2=\frac{N-2}{1-
  N}\frac{1}{{\mathbf{A}^{-1}}_{k,k}},$$
where $\mathbf{C}_{k,k}$ and ${\mathbf{A}^{-1}}_{k,k}$ are the $k$-th
element on the main diagonal of $\mathbf{C}$ and $\mathbf{A}^{-1}$,
respectively. 
\end{proposition}
\begin{proof}
Denote by $\mathbf{D}=\operatorname{diag}(\alpha_1,\ldots,\alpha_N)$,
then we have $\mathbf{A}=-\mathbf{D}+\mathbf{cc}^*$. Applying the
matrix inversion Lemma, we obtain 
\begin{align}
\mathbf{A}^{-1}&=-\mathbf{D}^{-1}-(\mathbf{D}^{-1}\mathbf{c})(1-\mathbf{c}^*\mathbf{D}^{-1}\mathbf{c})^{-1}(\mathbf{c}^*\mathbf{D}^{-1})\nonumber \\
&=
-\mathbf{D}^{-1}-(\mathbf{D}^{-1}\mathbf{c})(\mathbf{c}^*\mathbf{D}^{-1})\frac{1}{1-N}.\nonumber
\end{align}
We conclude the proof noticing that the $k$-th element on the main
diagonal of $\mathbf{A}^{-1}$ is a function of $N$ and $\alpha_k$,
\begin{align}
{\mathbf{A}^{-1}}_{k,k}=\frac{1}{\alpha_k}-\frac{1}{1-N}\frac{1}{\alpha_k}=\frac{N-2}{1-N}\frac{1}{\alpha_k}.\nonumber
\end{align}
\end{proof}

Once we have recovered $\mathbf{C}$, we obtain the coefficients
$\mathbf{c}$ by taking the eigenvector of $\mathbf{C}$ corresponding
to the largest eigenvalue.

\section*{Acknowledgment}
The authors would like to thank Prof. Gervais Chapuis for the
insightful discussions on X-ray crystallography.  This research was
supported by an ERC Advanced Grant -- Support for Frontier Research --
SPARSAM Nr: 247006.

\begin{IEEEbiography}[{\includegraphics[width=1in,height=1.25in,clip,keepaspectratio]{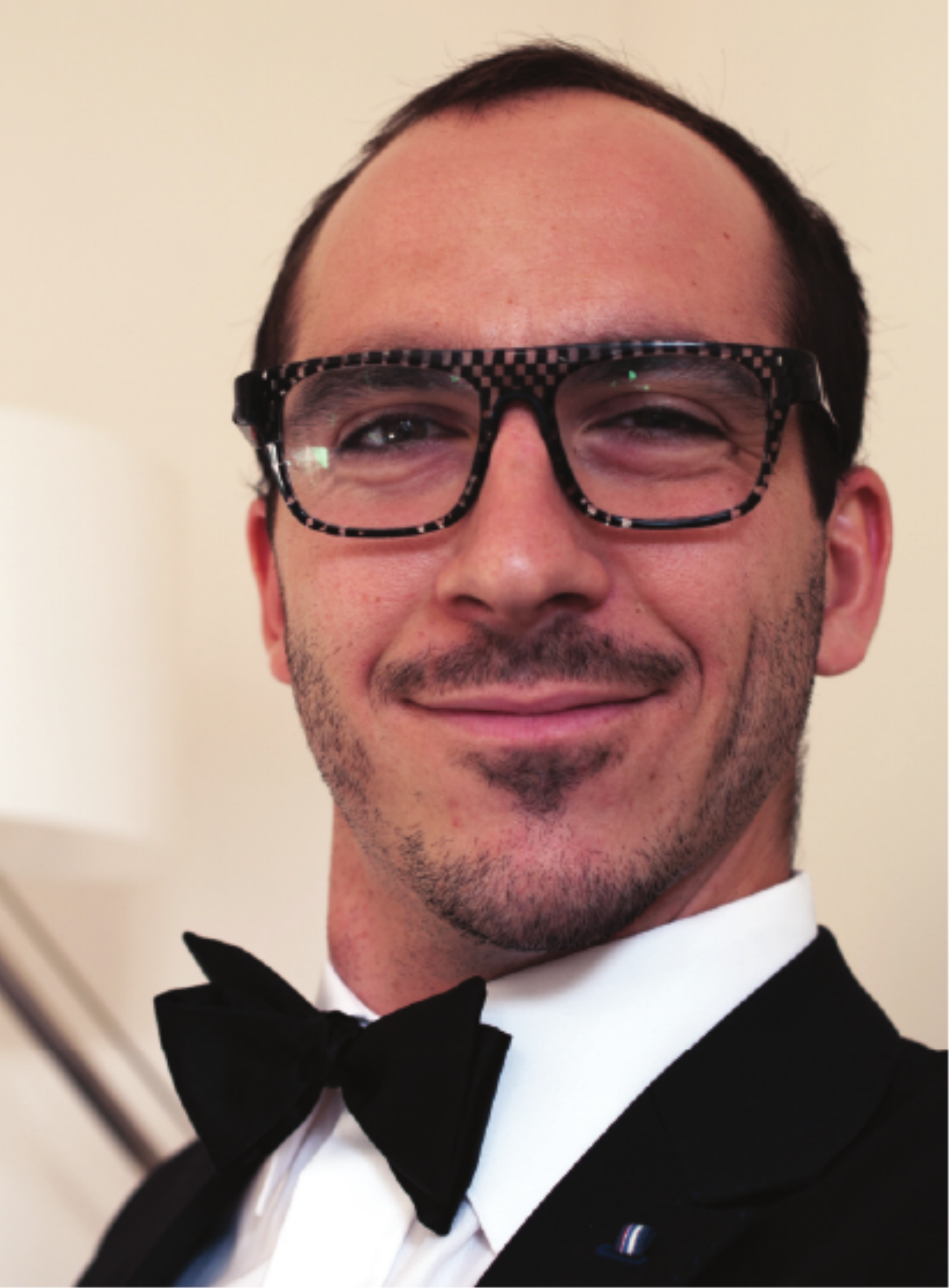}}]{Juri
    Ranieri} received both his M.S. and B.S. degree in Electronic
  Engineering in 2009 and 2007, respectively, from Universit{\'a} di
  Bologna, Italy. From July to December 2009, he joined as a visiting
  student the Audiovisual Communications Laboratory (LCAV) at EPF
  Lausanne, Switzerland. From January 2010 to August 2010, he was with
  IBM Zurich to investigate the lithographic process as a signal
  processing problem.  From September 2010, he is in the doctoral
  school at EPFL where he joined LCAV under the supervision of
  Prof. Martin Vetterli and Dr. Amina Chebira. From April 2013 to July
  2013, he was an intern at Lyric Labs of Analog Devices, Cambridge,
  USA. His main research interests are inverse problems of physical fields
  and the spectral factorization of autocorrelation functions.
\end{IEEEbiography}
\vspace{-1cm}
\begin{IEEEbiography}[{\includegraphics[width=1in,height=1.25in,clip,keepaspectratio]{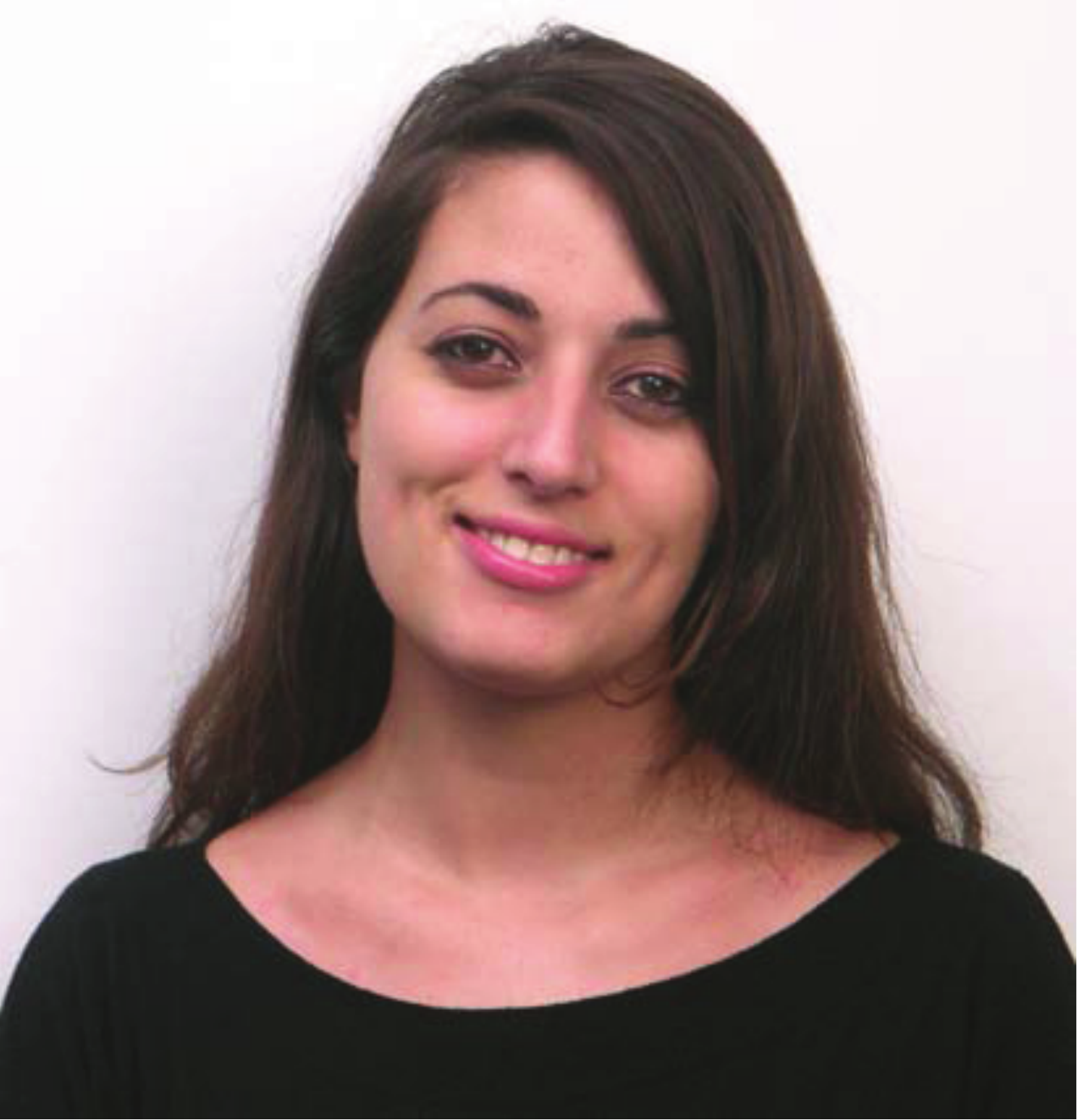}}]{Amina
    Chebira} is a senior research and development engineer at the
  Swiss Center for Electronics and Microtechnology (CSEM) in Neuch{\^
    a}tel, Switzerland. In 1998, she obtained a Bachelor degree in
  mathematics from University Paris 7 Denis Diderot. She received the
  B.S. and M.S. degrees in communication systems from the Ecole
  Polytechnique Fédérale de Lausanne (EPFL) in 2003 and the
  Ph.D. degree from the Biomedical Engineering Department, Carnegie
  Mellon University, Pittsburgh, PA, in 2008, for which she received
  the biomedical engineering research award. She then held a
  Postdoctoral Researcher position with the Audiovisual Communications
  Laboratory, EPFL, from 2008 to 2012. Her research interests include
  frame theory and design, biomedical signal and image processing, pattern
  recognition, filterbanks and multiresolution theory. 
\end{IEEEbiography}

\vspace{-1cm}

\begin{IEEEbiography}[{\includegraphics[width=1in,height=1.25in,clip,keepaspectratio]{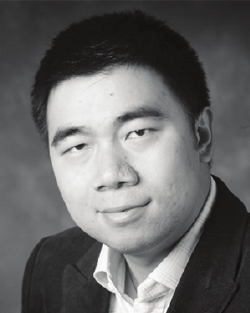}}]{Yue
    M. Lu} received the M.Sc. degree in mathematics and the
  Ph.D. degree in electrical engineering from the University of
  Illinois at Urbana-Champaign, Urbana, IL, both in 2007.  He was a
  Research Assistant at the University of Illinois at
  Urbana-Champaign, and was with Microsoft Research Asia, Beijing,
  China and Siemens Corporate Research, Princeton, NJ. From September
  2007 to September 2010, he was a postdoctoral researcher at the
  Audiovisual Communications Laboratory at Ecole Polytechnique
  Fédérale de Lausanne (EPFL), Switzerland. He is currently an
  Assistant Professor of electrical engineering at Harvard University,
  Cambridge, MA, directing the Signals, Information, and Networks
  Group (SING) at the School of Engineering and Applied Sciences. His
  research interests are in the general areas of signal processing,
  statistical inference and imaging.  He received the Most Innovative
  Paper Award of IEEE International Conference on Image Processing
  (ICIP) in 2006 for his paper (with Minh N. Do) on the construction
  of directional multiresolution image representations, and the Best
  Student Paper Award of IEEE ICIP in 2007. He also coauthored a paper
  (with Ivan Dokmanić and Martin Vetterli) that won the Best Student
  Paper Award of IEEE International Conference on Acoustics, Speech
  and Signal Processing in 2011.
\end{IEEEbiography}
\vspace{-1cm}
\begin{IEEEbiography}[{\includegraphics[width=1in,height=1.25in,clip,keepaspectratio]{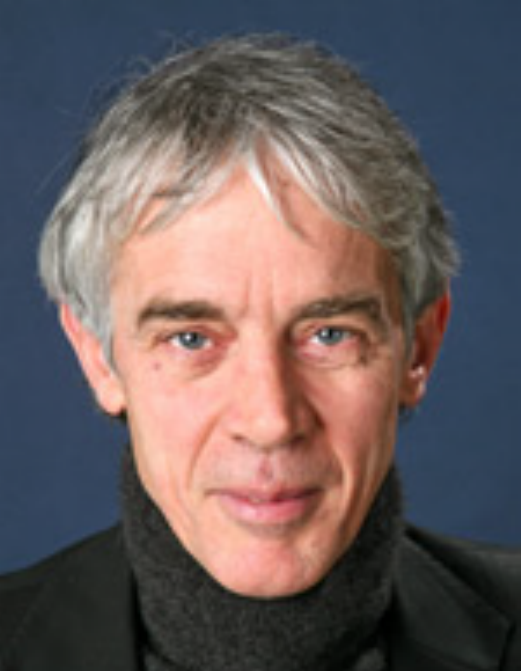}}]{Martin
    Vetterli} was born in 1957 and grew up near Neuchatel. He received
  the Dipl. El.-Ing. degree from Eidgenossische Technische Hochschule
  (ETHZ), Zurich, in 1981, the Master of Science degree from Stanford
  University in 1982, and the Doctorat es Sciences degree from the
  Ecole Polytechnique Federale, Lausanne, in 1986. After his
  dissertation, he was an Assistant and then Associate Professor in
  Electrical Engineering at Columbia University in New York, and in
  1993, he became an Associate and then Full Professor at the
  Department of Electrical Engineering and Computer Sciences at the
  University of California at Berkeley. In 1995, he joined the EPFL as
  a Full Professor. He held several positions at EPFL, including Chair
  of Communication Systems and founding director of the National
  Competence Center in Research on Mobile Information and
  Communication systems (NCCR-MICS). From 2004 to 2011 he was Vice
  President of EPFL and from March 2011 to December 2012, he was the
  Dean of the School of Computer and Communications Sciences. Since
  January 2013, he leads the Swiss National Science Foundation. He
  works in the areas of electrical engineering, computer sciences and
  applied mathematics. His work covers wavelet theory and
  applications, image and video compression, self-organized
  communications systems and sensor networks, as well as fast
  algorithms, and has led to about 150 journals papers. He is the
  co-author of three textbooks, with J. Kovacevic, ”Wavelets and
  Subband Coding” (PrenticeHall, 1995), with P. Prandoni, ”Signal
  Processing for Communications”, (CRC Press, 2008) and with
  J. Kovacevic and V. Goyal, of the forthcoming book ”Fourier and
  Wavelet Signal Processing” (2012). His research resulted also in
  about two dozen patents that led to technology transfers to
  high-tech companies and the creation of several start-ups. His work
  won him numerous prizes, like best paper awards from EURASIP in 1984
  and of the IEEE Signal Processing Society in 1991, 1996 and 2006,
  the Swiss National Latsis Prize in 1996, the SPIE Presidential award
  in 1999, the IEEE Signal Processing Technical Achievement Award in
  2001 and the IEEE Signal Processing Society Award in 2010. He is a
  Fellow of IEEE, of ACM and EURASIP, was a member of the Swiss
  Council on Science and Technology (2000-2004), and is an ISI highly
  cited researcher in engineering.
\end{IEEEbiography}

\bibliographystyle{IEEEtran} \bibliography{biblio.bib}

\end{document}